\documentclass[conference,letterpaper]{IEEEtran}
\usepackage{cite}
\usepackage{amsmath,amsthm,amssymb,amsfonts}
\usepackage{algorithmic}
\usepackage{graphicx}
\usepackage{epstopdf}
\usepackage{textcomp}
\usepackage{xcolor}
\usepackage{balance}
\usepackage{fancyhdr}
\footnoterule

\newtheorem {Proposition}{Proposition}

\def\BibTeX{{\rm B\kern-.05em{\sc i\kern-.025em b}\kern-.08em
    T\kern-.1667em\lower.7ex\hbox{E}\kern-.125emX}}
\begin{document}

\title{Mediumband Wireless Communication
\thanks{This work is jointly supported by Dublin City University, Republic of Ireland and The Royal Society of Edinburgh, UK.}
}

\author{\IEEEauthorblockN{Dushyantha A. Basnayaka, \textit{Senior Member}, \textit{IEEE}}
\IEEEauthorblockA{\textit{School of Electronics Engineering, Dublin City University}\\
Collins Ave, Dublin 9, Ireland\\
E-mail: d.basnayaka@dcu.ie}
}

\maketitle
\begin{abstract}
The fundamental phenomenon widely known as ``multipath'' is unavoidable in wireless communication, and affects almost every element of modern wireless communication systems. The impact of multipath on the received signal depends on whether the delay spread (i.e., spread of time 
delays associated with different multipath components) is large or small relative to the signalling period of the wireless communication system. In \textit{narrowband} systems, the delay spread is about one tenth (or less) of the signalling period. The delay spread and the signalling period of \textit{broadband} systems are in the same order of magnitude. In between these two extremes, there appears to exist an important, yet overlooked, class of systems whose delay spread is neither small nor large enough for them to fall into these two basic classes. In this paper, the effect of multipath on this class of systems denoted henceforth as ``\textit{mediumband}'' is studied, and its channel is characterized in compact form in order to enable future research into this class of wireless communication systems.     
\end{abstract}

\section{Introduction}
The global network of  telecommunication is a global effort spearheaded by the United Nations (UN), and is indispensable to the modern living. The role, that wireless communication plays in this global effort is widely evident\cite{ITU2022}. Wireless communication systems exploit different mediums like radio waves, infrared light and visible light to enable wireless data transmission. Its ability to keep people connected on the go and in emergencies are particularly appealing.\\
\indent In modern digital radio wireless communication, a transmitter (TX) sends data by transmitting a modulated electromagnetic (EM) wave, where typically the envelop of the EM wave varies according to a data signal. This data signal is typically an analog signal, but in digital wireless radio communication, only the signal points separated regularly in time carry information. The receiver (RX) receives an untidy mixture of attenuated and delayed versions (i.e., multipath) of the transmitted signal before the received signal being sampled and decoded. Nearly 50 years of wireless communication research efforts has been, and also is being, dedicated to perfect the detection of the desired signal from this untidy mixture of signals \cite{Stuber02}. The state-of-the-art of wireless communication is 5th-generation (5G) \cite{Jun14}.\\
\indent The multipath, that occurs due to the various objects in the environment between TX and RX, is both a blessing and an impairment for wireless communication. Often it is the multipath, that enables wireless communication when there is no line-of-sight (LoS) between TX and RX, which is the case in overwhelming number of instances of daily communication between people. On the other hand, in the presence of severe multipath, it is a considerably involving task to effectively and reliably detect the desired signal. The relative amplitudes and delays of these multipath components could combine constructively and in some instances destructively, and give rise to a concept known as ``fading". The strength and the nature of fading dictate almost every aspect of wireless systems, and is also the basis of the many common definitions of classes of wireless communication systems. \cite{Gold05}.\\
\subsection{Narrowband vs Broadband}
\begin{figure}[t]
	\centerline{\includegraphics*[scale=0.40]{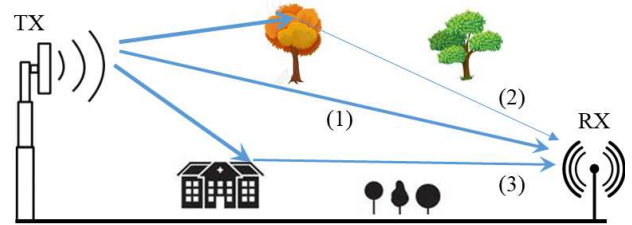}}
	\caption{A sketch of a typical radio wave communication system, where only three multipath components are shown. The arrows indicate the direction of propagation of the underlying EM waves.}\label{fig:fig1}
\end{figure}
\indent The classifications of narrowband and broadband (also known as wideband) are notable. A wireless system is said to be narrowband, if the symbol period (say $T_s$) is significantly greater than the delay spread (say $T_m$) of the multipath. The delay spread is significantly greater than $T_s$ in broadband wireless systems. The delay spread referred herein is the relative time difference between the latest and earliest significant multipath components\footnote{The delay spread is a random quantity, and often quantized by average delay spread or rms delay spread. These more technical definitions are omitted in this introductory section.}. The delay spread is dependent on the propagation environment, and in indoor settings, it typically ranges from 10 to 1000 nanoseconds, and in urban areas, delay spread can be as high as a few tens of microseconds. These differences in delay spread have many implications.\\
\indent If one wishes to have a relatively simpler narrowband system, the symbol period should be adjusted according to the delay spread, the longer the delay spread the longer the symbol period (in order at least to ensure $T_s \geq 10T_m$). This ensures the classical narrowband channel model, where the effect of multipath is modelled to a single multiplicative fading factor. The longer symbol period means low bit rates, which is typically undesirable. Also the narrowband systems designed to cope with smaller delay spread perform poorly in environments with large delay spread due to the excessive intersymbol interference (ISI). If one wishes to achieve higher bit rate in environments, where the delay spread is excessive, relatively complex wideband systems must be employed, where $T_m$ can be many times higher than the symbol period leading to a tapped delay line channel model \cite{Cassioli02}.\\ 
\indent The multipath is present in both classes of systems. The narrowband systems do not resolve multipath, while wideband systems do so up to some extent. \textbf{Since multipath is environment dependent, systems could encounter scenarios, either intermittently or otherwise, where the delay spread satisfies neither narrowband nor broadband constraints. For instance, when $T_s/10 \leq T_m \leq 9T_s/10$, the effect of multipath cannot be simply reduced to a single multiplicative factor, but still delay spread is sufficiently not wide enough (with respect to $T_s$) to resolve multipath.} The current paper analyses this class of radio wireless systems denoted henceforth as ``\textit{mediumband}'', and studies how the effect of multipath can be accurately captured into a channel model.       

\section{System Model}
\begin{figure*}
	\begin{align} \tag{11} \label{alpha:eq1_opt}
	\eta_o &= \sqrt{(1-0.25\beta)\left[\left(\sum_{n=0}^{N-1} |\gamma_n|^2 \right) - |h_o|^2 \right] + \sum_{n=0}^{N-1} \sum_{\substack{m=0 \\ m \neq n}}^{N-1} \gamma_n \gamma_m^* R(\tau_n-\tau_m)}.
	\end{align} 
\end{figure*}  
\begin{figure*}
	\begin{align} \tag{12} \label{auto-corr:eq1}
	R\left(\tau\right) = \left.\operatorname{sinc}\left( \frac{\tau}{T_s} \right) \frac{\cos\left( \beta \frac{\pi \tau}{T_s} \right)}{1 - \left( \frac{2 \beta \tau}{T_s} \right)^2} - \frac{\beta}{4} \operatorname{sinc}\left(\beta \frac{\tau}{T_s} \right) \frac{\cos\left( \frac{\pi \tau}{T_s} \right)}{1 - \left( \frac{\beta \tau}{T_s} \right)^2} \right.
	\end{align}
\end{figure*}
\begin{figure}[t]
	\centerline{\includegraphics[scale=0.6]{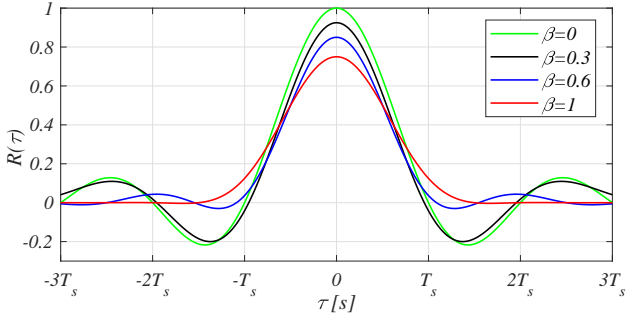}}
	\vspace{-3mm}
	\caption{$R(\tau)$ for different values of $\beta$, where $R(0)=1-0.25\beta$.}\label{fig:fig2}
\end{figure}
A radio wave (RF) wireless communication system with a single TX and a single RX in a rich scattering environment as shown in Fig. \ref{fig:fig1} is considered. Let $s(t)$ be the baseband equivalent transmitted signal corresponding to a single frame:
\begin{align} \label{eq1}
s(t) &= \sum_k I_k g(t-kT_s),
\end{align} 
where the symbol rate is $T_s$ and $\{I_k\}$ is the sequence of amplitudes drawn from a $2^b$ element constellation (e.g. BPSK, 4-PAM, 4-QAM, etc) by mapping $b$-bit binary blocks from an underlying information sequence $d_k$, and $g(t)$ is the pulse shaping filter. Combining the effects of root-raised-cosine transmit and receive pulse shaping filters, we herein assume a raised-cosine pulse for $g(t)$ with roll-off factor $\beta$. So $g(t)$ in time domain is given by:
\begin{align}
g(t) &= \begin{cases}
\frac{\pi}{4}\textup{sinc}\left(\frac{1}{2\beta}\right), & t \pm \frac{T_s}{2\beta} \\
\textup{sinc}\left(\frac{t}{T_s}\right)\frac{\cos\left(\frac{\pi \beta t}{T_s}\right)}{1-\left(\frac{2\beta t}{T_s}\right)^2}, & \text{otherwise.}
\end{cases}
\end{align}
In this paper, for simplicity, we assume only one dimensional constellations like BPSK and 4-PAM resulting real amplitudes for $\{I_k\}$, which in turn result in real $s(t)$. The transmitted RF signal is mathematically given by:
\begin{align}\label{rf:eq1}
x(t)= \text{Re}\left\{\sqrt{E_s}s(t)e^{j2\pi F_c t}\right\},
\end{align}
where $F_c$ is the carrier frequency, and $E_s$ is a factor that controls the transmit power. We also assume that the sequence $\{I_k\}$ is normalized such that $\mathcal{E}\left\{|I_k|^2\right\}=1$, which due to the effect of raised-cosine pulse in turn results:
\begin{align}\label{rf:eq2}
\mathcal{E}\left\{|s(t)|^2\right\}=1-\frac{\beta}{4} 
\end{align}
Also \eqref{rf:eq1} ensures that $\mathcal{E}\left\{|x(t)|^2\right\}=0.5E_s\left(1-\frac{\beta}{4} \right)$. The received RF signal at the RX, $y(t)$ can be given as a sum of multipath components as:
\begin{align}
y(t) &= \sqrt{E_s} \sum_{n=0}^{N-1} \text{Re}\left\{\alpha_n s(t-\tau_n)e^{j2\pi F_c (t-\tau_n)}\right\}, 
\end{align}
where $N$ is the number of multipath components, and $\tau_n$ and $\alpha_n$ are the absolute time delay and the gain of the $n$th component. In the absence of noise, the received baseband equivalent signal, $r(t)$ can be given by:
\begin{align}\label{eq:channel:1}
r(t) &= \sqrt{E_s} \sum_{n=0}^{N-1} \alpha_n e^{-j2\pi F_c\tau_n} s(t-\tau_n), \nonumber\\
&=  \sqrt{E_s} \sum_{n=0}^{N-1} \alpha_n e^{-j\phi_n} s(t-\tau_n),
\end{align}
where $\phi_n$ is known as the phase of the $n$th component, which is assumed as fixed at least within the time duration of a single frame corresponding to a case of static or terminals with slow relative movement. Without loss of generality, $0$ is assumed to be the path index of the earliest path meaning $\alpha_0$ and $\tau_0$ are the path gain and absolute delay of the earliest (also the shortest) path. Also, let the delay spread be defined as $T_m=\max_n |\tau_n-\tau_0|$.
In narrowband channels, where the constraint $T_m \leq 0.1T_s$ is at least approximately satisfied, the channel input-output (IO) relationship in \eqref{eq:channel:1} reduces to:
\begin{align}\label{eq:channel:2}
r(t) &\approx \sqrt{E_s} \left(\sum_{n=0}^{N-1} \gamma_n \right) s(t-\hat{\tau}),
\end{align}
where $\gamma_n=\alpha_ne^{-j\phi_n}$ for $\forall n$. Here, the multipath components are said to be nonresolvable, and combined into a single multipath component with delay $\hat{\tau} \approx \tau_0 \approx \tau_1 \dots \approx \tau_{N-1}$. The symbol timing synchronizer at the RX typically synchronizes to this common delay, $\hat{\tau}$ \cite{Coulson01}.\\
\indent In the mediumband regime, it can be seen that, even though multipath components are nonresolvable, the IO relationship in \eqref{eq:channel:2} is no longer accurate. The simulation studies (see Sec. \ref{simulation_study}) show that, as $T_m$ increases beyond $0.1T_s$, the narrowband identity in \eqref{eq:channel:2} weakens gradually. Hence, using the narrowband identity in \eqref{eq:channel:2} as the basis, in the sequel, we derive a new characterization for mediumband channels.    
\section{Channel Characterization}
\begin{Proposition}\label{proposition1}
In mediumband channels, where the delay spread  satisfies neither narrowband nor wideband assumptions, the baseband equivalent received signal in the absence of noise can be accurately modelled as:
\begin{align}\label{eq:r(t)}
r(t) &= \sqrt{E_s}h_os(t-\hat{\tau})+\sqrt{E_s} \eta_o u(t),
\end{align}
and the baseband equivalent received signal in the presence of noise as:
\begin{align}\label{eq:r2(t)}
r'(t) &= r(t) + n(t),
\end{align} 
where $s(t-\hat{\tau})$ is the desired signal; $u(t)$ is a complex uncorrelated zero mean unit variance interfering signal; and $n(t)$ is a complex zero mean additive-white-Gaussian-noise (AWGN) signal with $\sigma^2$ variance. The fading coefficients, $h_o$ and $\eta_o$ are respectively given by:
\begin{align}\label{eq:H}
h_o &= \frac{\sum_{n=0}^{N-1} \gamma_n R(\tau_n-\hat{\tau})}{1-\frac{\beta}{4}},
\end{align}
and \eqref{alpha:eq1_opt}, where $R(\tau)=\mathcal{E}\left\{s(t)s(t+\tau)\right\}$ is the autocorrelation function of $s(t)$, which is given in \eqref{auto-corr:eq1} and shown in Fig. \ref{fig:fig2}. The noise variance $\sigma^2$ is not dependent on the fading parameters, but only dependent on the noise bandwidth of the pulse shaping filter at the RX and the spectral density of the thermal noise, $N_0$. In room temperature, $N_0=-174$ dBm/Hz.  
\end{Proposition}

\begin{proof}
	See Appendix A
\end{proof}
The channel model in \eqref{eq:r(t)} and \eqref{eq:r2(t)} shows the desired, interfering and noise signals  in additive form, and the effect of multipath fading is conveniently captured as multiplicative factors. The model also captures the effects of transmit power (through $E_s$), pulse shaping (through $\beta$), and modulation (through $R(\tau)$). In the sequel, a few notable observations about the proposed mediumband channel model are made.   
\subsection{Signal-to-Interference-plus-Noise-Ratio (SINR)}
Signal-to-interference-plus-noise-ratio (SINR) is an important quantity, that captures the error and data rate performance of wireless communication systems. From \eqref{eq:r2(t)}, SINR of mediumband wireless channel can be obtained in compact form as:
\begin{align}
\setcounter{equation}{12}
\text{Mediumband:   \qquad} \text{SINR} &= 	\frac{E_s|h_o|^2\left(1-\frac{\beta}{4}\right)}{E_s\eta_o^2 + \sigma^2}, 
\end{align}
where the identity in \eqref{rf:eq2} is used. This clearly shows the dependence of SINR on fading and also on the spectral properties of pulse shaping filters and AWGN noise.  
\subsection{Narrowband Channel As a Special Case}
Furthermore, it can be seen that the narrowband channel in \eqref{eq:channel:2} is a special case of the mediumband channel. Consider $h_o$ and $\eta_o$ as $T_m=\max_n |\tau_n-\tau_0| \rightarrow 0$. As $T_m \rightarrow 0$, all delay deferences (i.e., $\tau_n - \tau_m$ and $\tau_n - \hat{\tau}$ $\forall n,m$) also approach zero. Therefore, \eqref{auto-corr:eq1} yields:
\begin{subequations}
\begin{equation}\label{R:eqs}
\lim_{|\tau_n - \hat{\tau}| \rightarrow 0} R(\tau_n - \hat{\tau}) = 1-\frac{\beta}{4}, \qquad \forall n
\end{equation}
\begin{equation}
\lim_{|\tau_n - \tau_m| \rightarrow 0} R(\tau_n - \tau_m) = 1-\frac{\beta}{4}, \qquad \forall n,m 
\end{equation}
\end{subequations}
Substituting these limits into equations for $h_o$ and $\eta_o$ in \eqref{eq:H} and \eqref{alpha:eq1_opt}, it can be shown that:
\begin{subequations}\label{H_eta:eqs}
	\begin{equation}
	\lim_{\substack{|\tau_n - \hat{\tau}| \rightarrow 0 \\ \forall n}} h_o = \sum_{n=0}^{N-1} \gamma_n =  \sum_{n=0}^{N-1} \alpha_n e^{-j\phi_n},
	\end{equation}
	\begin{equation}
	\lim_{\substack{R(\tau_n - \tau_m) \rightarrow \left(1-\frac{\beta}{4}\right) \\ \forall n,m}} \eta_o = 0,
	\end{equation}
\end{subequations}
which ensure the convergence of the mediumband channel into a narrowband channel as $T_m \rightarrow 0$.
\subsection{Cross-correlation of Desired and Interfering Signals}
Consider the cross-correlation of the desired signal, $h_os(t-\hat{\tau})$ and the interfering signal, $\eta_o u(t)$: $\mathcal{E}\left\{\overline{h_os(t-\hat{\tau})}\eta_o u(t)\right\}$, where the ``overline'' denotes the complex conjugation. The desired cross-correlation can equivalently be expressed by:
\begin{align}\label{cross_corr:eq1}
\mathcal{L} &= \mathcal{E}\left\{\overline{h_os(t-\hat{\tau})}\eta_o u(t)\right\}, \nonumber \\
&= \mathcal{E}\left\{\overline{\sqrt{E_s}h_os(t-\hat{\tau})}\left(r(t)- \sqrt{E_s}h_os(t-\hat{\tau})\right)\right\}.
\end{align}
The expectation operation with respect to the random process, $s(t)$, yields:
\begin{align}\label{cross_corr:eq2}
\!\!\!\mathcal{L} &= E_s \left[h_o^* \left(\sum_{n=0}^{N-1} \gamma_n R(\tau_n - \hat{\tau})\right)- |h_o|^2\left(1-\frac{\beta}{4}\right)\right].
\end{align}   
From \eqref{eq:H}, it is known that $\sum_{n=0}^{N-1}\gamma_n R(\tau_n - \bar{\tau})=h_o(1-\frac{\beta}{4})$, which in turn confirms that if $h_o$ is chosen optimally as given in \eqref{eq:H}, the correlation between the desired signal and the interfering signal will be zero.   
%
\section{Simulation Study}\label{simulation_study}
In this section, a generic mediumband channel described in \eqref{eq:channel:1} is simulated on MATLAB. Typically $\tau_n$s are dependent on the environment, but without loss of generality, we assume $\tau_0=0$. The other delays, $\tau_n$ for $n=1,\dots,N-1$ are drawn from an uniform distribution, $U[0,T_m]$, where $T_m$ is the delay spread. The sequence of amplitudes are drawn from a BPSK constellation, so $\{I_k\} \in \left\{-1, 1\right\}$ for $\forall k$. Furthermore, the phases are drawn from a uniform distribution, $\phi_n \sim U[0,2\pi]$, and two scenarios for amplitudes namely ``uniform'' and ``exponential'' are considered. In the scenario of ``uniform'', $\alpha_1=\alpha_2=\dots =\alpha_n \propto 1/\sqrt{N}$. In the ``exponential'' scenario, the amplitudes are chosen such that  $\alpha_n \propto e^{-\kappa n}$, where $\kappa$ captures the decay of amplitudes. In both scenarios, the amplitudes are also normalized such that $\sum \alpha_n^2=1$. We consider the signal-to-interference ratio (SIR) defined by:
\begin{align}\label{seciv:eq1}
\text{SIR} &= \frac{\mathcal{E}\left\{\left|\sqrt{E_s}h_o s(t-\hat{\tau})\right|^2 \right\}}{\mathcal{E}\left\{\left|r(t) - \sqrt{E_s}h_o s(t-\hat{\tau})\right|^2 \right\}}, 
\end{align}
as the metric to assess the performance, where the expectation is over both $s(t)$ and fading. Due to the independence of the fading process and the channel input process, $s(t)$, the expectations in \eqref{seciv:eq1} are evaluated in two steps. The expectation over input process is firstly evaluated while keeping the fading parameters, $\Gamma=\left\{\gamma_n\right\}$, fixed, where $\hat{\tau}$ is chosen optimally using exhaustive search, and $h_o$ is from Proposition 1. Secondly, the expectation over the fading process is evaluated by repeating the first step until sufficiently stable values are obtained.\\
Fig. \ref{fig:fig3} shows the SIR performance of a mediumband wireless channel for different percentage delay spreads in both uniform and exponential amplitude profiles, where the percentage delay spread is defined by:
\begin{align}
\text{Percentage Delay Spread} &= \left(\frac{T_m}{T_s}\right) \text{ x } 100 \%.
\end{align}
It can be seen clearly that SIR gradually decreases as $T_m$ increases, but the decrease is greater in ``uniform'' amplitude profile. Furthermore, as shown on Fig. \ref{fig:fig4}, as $N$ increases, SIR performance increases, but quickly saturates (see $N>20$ on Fig. \ref{fig:fig4}). A detailed simulation study and experimental results are available on \cite{Bas2022}. 
\begin{figure}[t]
	\centerline{\includegraphics[scale=0.65]{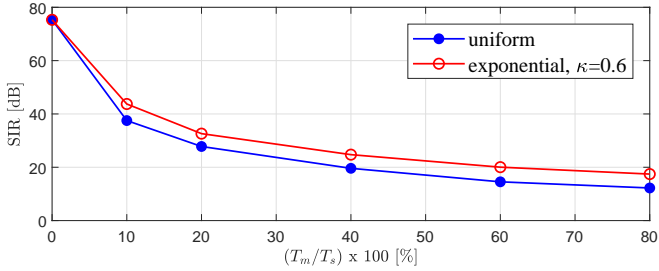}}
	\caption{\text{SIR} vs Percentage Delay Spread of mediumband wireless channels, where $\beta=0.8$ and $N=5$.}\label{fig:fig3}
\end{figure}
\begin{figure}[t]
	\centerline{\includegraphics[scale=0.65]{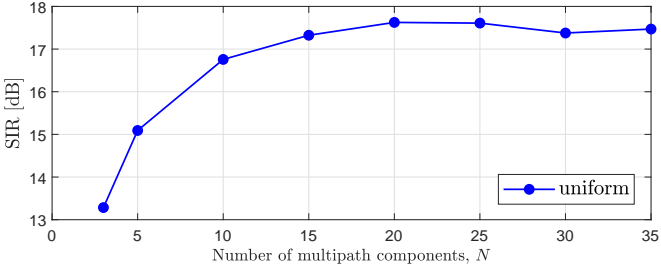}}
	\caption{\text{SIR} vs $N$ of mediumband wireless channels, where $\beta=0.8$ and percentage delay spread is $60\%$.}\label{fig:fig4}
\end{figure}
\balance
\section{Conclusion}
In this paper, a class of wireless channels denoted herein as ``mediumband'' was considered. In mediumband channels, the delay spread is neither small nor large enough, so the conventional models for narrowband and broadband channels appear to be not sufficient to capture the essence of these channels. We analytically studied the effect of different fading parameters (i.e., path delays, amplitudes and phases) on the mediumband channels, and proposed a channel characterization in compact form, where the effects of transmit power, modulation, pulse shaping and fading have been captured in convenient form.       
\appendices
\section{Proof of Proposition \ref{proposition1}}\label{Appendix:A}
\begin{figure*}
	\begin{align}\label{appA:eq3}
	\!\!\! J &= E_s \left[|h|^2 \mathcal{E}\left\{s(t-\hat{\tau})^2\right\} + \sum_{n=0}^{N-1} |\gamma_n|^2 \mathcal{E}\left\{s(t-\tau_n)^2\right\} +
	\sum_{n=0}^{N-1} \sum_{\substack{m=0 \\ m \neq n}}^{N-1} \gamma_n \gamma_m^*\mathcal{E}\left\{s(t-\tau_n)s(t-\tau_m)\right\}
	\right. \nonumber \\
	& \qquad \qquad \qquad \qquad \qquad \qquad \qquad \qquad \qquad \qquad \qquad \qquad \qquad \qquad - \left. \sum_{n=0}^{N-1} 2\text{Re}\left\{\gamma_n^*h\right\} \mathcal{E}\left\{s(t-\tau_n)s(t-\hat{\tau})\right\} \right] \tag{23}
	\end{align}
\end{figure*}
\begin{figure*}
	\begin{align}\tag{26} \label{alpha:eq1}
	J &= E_s \left[(1-0.25\beta)\left(|h|^2 + \sum_{n=0}^{N-1} |\gamma_n|^2 \right) + \sum_{n=0}^{N-1} \sum_{\substack{m=0 \\ m \neq n}}^{N-1} \gamma_n \gamma_m^* R(\tau_n-\tau_m)
	- \sum_{n=0}^{N-1} 2\text{Re}\left\{\gamma_n^*h\right\} R(\tau_n-\hat{\tau})\right].
	\end{align}
\end{figure*}
Assuming $\hat{\tau}$ is the time which the demodulator is synchronized to, we consider the error signal defined by:
\begin{align}\label{appA:eq1}
e(t) &= r(t) - \sqrt{E_s}hs(t-\hat{\tau}).
\end{align}
The analysis herein is valid whether the symbol timing synchronization, $\hat{\tau}$ in \eqref{appA:eq1} is optimal or not. If $\hat{\tau}$ is optimal, the error signal would be weaker. It is otherwise, if $\hat{\tau}$ is not optimal. However, in the simulation study in Sec. \ref{simulation_study}, we use exhaustive search to find the optimum timing for $\hat{\tau}$. Consider the following conditional expectation:
\begin{align}
J &= \mathcal{E}\left\{\left.|e(t)|^2\right | \Gamma, h \right\},\\
&= \mathcal{E}\left\{\left.\left|r(t) - \sqrt{E_s}hs(t-\hat{\tau})\right|^2 \right| \Gamma, h \right\}, \label{appA:eq2}
\end{align}
where the expectation is taken with respect to the random process, $s(t)$, and is also conditioned on the fading parameters, that is $\Gamma=\left\{\gamma_n\right\}$. This conditional expectation is the appropriate choice to characterize the behaviour of the channel for a given set of fading parameters. The \eqref{appA:eq2} can be expanded and simplified as shown in \eqref{appA:eq3}. 
It is clear that $\mathcal{E}\left\{s(t-\tau_n)s(t-\tau_m)\right\}$ appears in \eqref{appA:eq3} is the autocorrelation function of $s(t)$, which is a signal resulted from linear modulation. From \cite[eq: 4-4-11]{Proakis00}, the desired autocorrelation can be obtained as:
\begin{align*}
\mathcal{E}\left\{s(t)s(t+\tau)\right\} &= \frac{1}{T_s} \sum_{q=-\infty}^{\infty} \psi_{ii}(q) \psi_{gg}(\tau - qT_s), 
\end{align*}
where $\psi_{ii}(q)$ is the autocorrelation of the real information sequence $\left\{I_k\right\}$, and  $\psi_{gg}(\tau)$ is the time autocorrelation function of the raised cosine pulse, $g(t)$.  $\psi_{ii}(q)$ is defined as $\mathcal{E}_I\left\{I_k I_{k+q}\right\}$. In light of $\mathcal{E}\left\{|I_k|^2\right\}=1$, $\psi_{ii}(q)$ can be shown to be equal to:
\begin{align}
\setcounter{equation}{23}
\psi_{ii}(q) &= \begin{cases}
1 & q=0 \\
0 , & \text{otherwise.}
\end{cases}
\end{align}
Hence, $\mathcal{E}\left\{s(t)s(t+\tau)\right\} = \frac{1}{T_s} \psi_{gg}(\tau)$. Here the time autocorrelation function, $\psi_{gg}(\tau)$ is defined as $\mathcal{E}_t\left\{g(t)g(t+\tau)\right\}$, which is well known, and the desired result can be obtained as:
\begin{align}\label{appA:eq4}
\mathcal{E}\left\{s(t)s(t+\tau)\right\} &= R(\tau),
\end{align}
where $R(\tau)$ is given in \eqref{auto-corr:eq1}. Applying the result in \eqref{appA:eq4} into \eqref{appA:eq3}, $J$ can be simplified to get \eqref{alpha:eq1}. 
Due to the fact that the constellations are typically symmetric, $\mathcal{E}_I\left\{I_k\right\}=0$. Thus the linear digital modulation (i.e., \eqref{eq1}) ensures that:
\begin{align}\label{appA:eq5}
\setcounter{equation}{26}
\mathcal{E}\left\{s(t)\right\} &= 0.
\end{align}
So, the error process, $e(t)$ is also a zero mean random process with variance $J$, which is given in \eqref{alpha:eq1}.
\subsection{Optimization of $h$}
It is clear that $J$, which is quadratically dependent on $h$ can further be optimized. Consider the following partial derivatives of $J$ with respect to $h$:
\begin{align}
\frac{\partial J}{\partial h_I} = \left(1-\frac{\beta}{4}\right)h_I - \sum_{n=0}^{N-1}\text{Re}\left\{\gamma_n\right\}R(\tau_n-\hat{\tau}),
\end{align}
\begin{align}
\frac{\partial J}{\partial h_Q} = \left(1-\frac{\beta}{4}\right)h_Q - \sum_{n=0}^{N-1}\text{Im}\left\{\gamma_n\right\}R(\tau_n-\hat{\tau}).
\end{align}
where $h_I$ and $h_Q$ denote the real and imaginary parts of $h$ meaning $h=h_I+jh_Q$. Equating these partial derivatives to zero, which is:
\begin{align}
\frac{\partial J}{\partial h_I} = \frac{\partial J}{\partial h_Q} &= 0,
\end{align}
yields the optimum $h=h_o$, which can be found to be equal to:
\begin{align} \label{opt_h}
h_o &= \frac{\sum_{n=0}^{N-1} \gamma_n R(\tau_n-\hat{\tau})}{1-\frac{\beta}{4}}.
\end{align}
This completes the proof for $h_o$. Furthermore, by substituting the optimum value for $h$ in \eqref{opt_h}, one can obtain the optimum error variance of the error process, $e(t)$ as $E_s \eta_o^2$, where $\eta_o$ is in \eqref{alpha:eq1_opt}. Consequently, in statistically equivalent form, $e(t)$ may be modelled by $e(t) = \sqrt{E_s}\eta_o u(t)$,
where $u(t)$ is a zero mean, unit variance complex random process. In the absence of AWGN noise, $r(t)$ can thus be written as $r(t) = \sqrt{E_s}h_o s(t-\hat{\tau}) + \sqrt{E_s}\eta_o u(t)$.
\end{document}